\documentclass[reqno,11pt]{amsart}
\usepackage{color}
\IfFileExists{mymtpro2.sty}{%
  \usepackage[subscriptcorrection]{mymtpro2}
}{}

\usepackage{a4,amssymb}
\usepackage{graphicx}

\newtheorem{theorem}{Theorem}

\newtheorem{thm}{Theorem}[section]
\newtheorem{lemma}{Lemma}[section]

\newtheorem{pr}{Proposition}[section]
\theoremstyle{definition}

\newcommand{\bel}{\begin{equation} \label}
\newcommand{\ee}{\end{equation}}

\theoremstyle{remark}
\newtheorem{remark}[theorem]{Remark}
\newtheorem{myremarks}[theorem]{Remarks}

    {\end{nummer}\end{myremarks}}

\newcounter{numcount}
\newcommand{\labelnummer}{\mbox{\normalfont (\roman{numcount})}}%

\makeatletter

\newenvironment{nummer}%
  {\let\curlabelspeicher\@currentlabel%
    \begin{list}{\labelnummer}%
      {\usecounter{numcount}\leftmargin0pt%
        \topsep0.5ex\partopsep2ex\parsep0pt\itemsep0ex\@plus1\p@%
        \labelwidth2.5em\itemindent3.5em\labelsep1em%
      }%
    \let\saveitem\item%
    \def\item{\saveitem%
      \def\@currentlabel{{\upshape\curlabelspeicher}$\,$\labelnummer}}%
    \let\savelabel\label%
    \def\label##1{\savelabel{##1}%
      \@bsphack%
        \ifmmode\else%
          \protected@write\@auxout{}%
          {\string\newlabel{##1item}{{\labelnummer}{\thepage}}}%
        \fi%
      \@esphack%
    }%
  }{\end{list}}%

\renewcommand{\appendix}{\def\thesection{\textsc{Appendix}}}

 \let\leq\le
 \let\geq\ge

\DeclareMathOperator{\tr}{tr\kern1pt}

\newcommand{\Ran}{\mathop\mathrm{Ran}\nolimits}

\newcommand{\N}{{\mathbb N}}

\makeatletter

\newif\ifper\pertrue
\def\per{.}

\def\bti{\@ifnextchar[\bbti\bbbti}
\def\bbti[#1]#2{#2, #1.}
\def\bbbti#1{#1.}

\def\z{\@ifnextchar[\zz\zzz}
\def\zz[#1]#2#3#4#5{\perfalse\emph{#2} \textbf{#3}, #4 (#5) [#1]}
\def\zzz#1#2#3#4{\emph{#1} \textbf{#2}, #3 (#4)\ifper\per\fi\pertrue}

\def\pub{\@ifstar\pubstar\pubnostar}
\def\pubnostar{\@ifnextchar[\@@pubnostar\@pubnostar}
\def\@@pubnostar[#1]#2#3#4{#2, #3, #4, #1\ifper\per\fi\pertrue}
\def\@pubnostar#1#2#3{#1, #2, #3\ifper\per\fi\pertrue}
\def\pubstar[#1]#2#3#4{\perfalse #2, #3, #4 [#1]\pertrue}

\makeatother

\newcommand{\beq}{\begin{equation}}
\newcommand{\eeq}{\end{equation}}
\newcommand{\ba}{\begin{array}}
\newcommand{\ea}{\end{array}}
\newcommand{\bea}{\begin{eqnarray}}
\newcommand{\eea}{\end{eqnarray}}
\newcommand{\beas}{\begin{eqnarray*}}
\newcommand{\eeas}{\end{eqnarray*}}

{

\newcommand{\R}{\mathbb{R}}

\newcommand{\Z}{\mathbb{Z}}

\newcommand{\C}{\mathbb{C}}

\def\P{I\kern-.30em{P}}
\def\E{I\kern-.30em{E}}
\renewcommand{\E}{\mathbb{E}\mkern2mu}
\renewcommand{\P}{\mathbb{P}}

\begin{document}

\title[Decorrelation estimates for non rank one perturbations]{Decorrelation estimates for random Schr\"odinger operators with non rank one perturbations}

\author[P.\ D.\ Hislop]{Peter D.\ Hislop}
\address{Department of Mathematics,
    University of Kentucky,
    Lexington, Kentucky  40506-0027, USA}
\email{peter.hislop@uky.edu}

\author[M.\ Krishna]{M.\ Krishna}
\address{Institute for Mathematical Sciences,
IV Cross Road, CIT Campus,
 Taramani, Chennai 600 113,
 Tamil Nadu, India}
\email{krishna@imsc.res.in}

\thanks{PDH was partially supported by NSF through grant DMS-1103104.
MK was partially supported by IMSc Project 12-R\&D-IMS-5.01-0106. The authors thank F.\ Klopp and C.\ Shirley for discussions on eigenvalue statistics and decorrelation estimates.}

%\thanks{Version of \today}

\begin{abstract}
   We prove decorrelation estimates for generalized lattice Anderson models on $\Z^d$ constructed with finite-rank perturbations in the spirit of Klopp \cite{klopp}. These are applied to prove that the local eigenvalue statistics $\xi^\omega_{E}$ and $\xi^\omega_{E^\prime}$, associated with two energies $E$ and $E'$ satisfying $|E - E'| > 4d$, are independent. That is, if $I,J$ are two bounded intervals, the random variables $\xi^\omega_{E}(I)$ and $\xi^\omega_{E'}(J)$, are independent and distributed according to a compound Poisson distribution whose L\'evy measure has finite support. We also prove that the extended Minami estimate implies that the eigenvalues in the localization region have multiplicity at most the rank of the perturbation. % in dimensions $d \geq 2$.
 \end{abstract}

\maketitle \thispagestyle{empty}

%\newpage

%\today

\tableofcontents

%\vspace{.2in}


%%%%%%%%%%%%%%%%%%%%%%%%%%%%%%%%%%%%%%%%%%%%%%%%%%%%%%%%%%%%%%%%%%%%%%%%%%%%%%%%%%%%%%%%%%%%%%%%%%%%%%%%%%%

\section{Statement of the problem and results}\label{sec:introduction}
\setcounter{equation}{0}

We consider random Schr\"odinger operators $H^\omega = {\mathcal L} + V_\omega$ on the lattice Hilbert space $\ell^2 (\Z^d)$ (or, for matrix-valued potentials, on $\ell^2 (\Z^d) \otimes \C^{m_k}$),
and prove that certain natural random variables
associated with the local eigenvalue statistics around two distinct energies $E$ and $E^\prime$, in the region of complete localization $\Sigma_{\rm CL}$ and with $| E - E^\prime | > 4d$, are independent. From previous work \cite{hislop-krishna1}, these random variables
distributed according to a compound Poisson distribution.
The operator ${\mathcal L}$ is the discrete Laplacian on $\Z^d$, although this can be generalized.
For these lattice models, the random potential $V_\omega$ has the form
\beq\label{eq:potential1}
(V_\omega f)(j) = \sum_{i \in {\mathcal J}} \omega_i  (P_i f)(j),
\eeq   
where $\{ P_i \}_{i \in {\mathcal{J}}}$ is a family of finite-rank projections with the same rank $m_k \geq 1$, the set $\mathcal{J}$ is a sublattice of $\Z^d$, and $\sum_{i \in {\mathcal{J}}} P_i = I$. We assume that $P_i = U_i P_0 U_i^{-1}$, for $i \in \mathcal{J}$,
where $U_i$ is the unitary implementation of the translation group
$(U_i f)(k) = f(k+i)$, for $i,k \in \Z^d$. 
%so that $\sum_{i \in {\mathcal{J}}} P_i = I$.
%where $\{ P_i \}_{i \in {\mathcal{J}}}$ is a family of finite-rank projections with the same rank $m_k \geq 1$ and
%so that $\sum_{i \in {\mathcal{J}}} P_i = I$.
The coefficients $\{ \omega_i \}$ are a family of independent, identically distributed (iid) random variables with a bounded density of compact support on a product probability space $\Omega$ with probability measure $\P$.
It follows from the conditions above that the family of random Schr\"odinger operators $H^\omega$ is ergodic with respect to the translations generated by $\mathcal{J}$.
%.  If, for example, the family of projections $P_i$ is translation covariant, so that $U_k P_i U_k^{-1} = P_{k+i}$, %for all $k,i \in \mathcal{I}$, then the family is ergodic with respect to these translations.

One example on the lattice is the polymer model. For this model, the projector $P_i = \chi_{\Lambda_k(i)}$ is the characteristic function
on the cube $\Lambda_k(i)$ of side length $k$ centered at $i \in \Z^d$. The rank of $P_i$ is $(k+1)^d$
and the set $\mathcal{J}$ is chosen so that $\cup_{i \in \mathcal{J}} {\Lambda_k(i)} = \Z^d$.
Another example is a matrix-valued model for which $P_i$, $i \in \Z^d$, projects onto an $m_k$-dimensional subspace,
 and $\mathcal{J} = \Z^d$. The corresponding Schr\"odinger operator is
\beq\label{eq:model1}
H^\omega = {\mathcal L} +  \sum_{i \in {\mathcal{J}}} \omega_i  P_i ,
\eeq
where ${\mathcal L}$ is the discrete lattice Laplacian $\Delta$ on $\ell^2 (\Z^d)$,
or $\Delta \otimes I$ on $\ell^2(\Z^d) \otimes \C^{m_k}$ (or, more generally, $\Delta \otimes A$,
where $A$ is a nonsingular $m_k \times m_k$ matrix), respectively.
%The family $\{ \omega_i \}$ is a family of independent, identically distributed (iid)
In the following, we denote by $H_{\omega, \ell}$ (or simply as
$H_\ell$ omitting the $\omega$) the matrices $\chi_{\Lambda_\ell}H^\omega\chi_{\Lambda_\ell}$ and similarly $H_{\omega,L}, H_L$ by replacing $\ell$ with $L$, for positive integers $\ell$ and $L$.

A lot is known about the eigenvalue statistics for random Schr\"odinger operators on $\ell^2 (\R^d)$.
When the projectors $P_i$ are rank one projectors, the local eigenvalue statistics in the localization regime has been proved to be given by a Poisson process by Minami \cite{minami1} (see also Molchanov \cite{molchanov} for a model on $\R$ and Germinet-Klopp \cite{germinet-klopp} for a comprehensive discussion and additional results). For the non rank one case, Tautenhahn and Veseli\'c \cite{tautenhahn-veselic} proved a Minami estimate for certain models that may be described as weak perturbations of the rank one case. The general non finite rank case was studied by the authors in \cite{hislop-krishna1} who proved that, roughly speaking, the local eigenvalue statistics are compound Poisson. This result also holds for random Schr\"odinger operators on $\R^d$.

In this paper, we further refine these results for lattice models with non rank one projections and prove, roughly speaking, that the processes associated with two distinct energies are independent. Klopp \cite{klopp} proved decorrelation estimates for lattice models in any dimension. He applied them to show that the local eigenvalue point processes at distinct energies converge to independent Poisson processes (in dimensions $d >1$ the energies need to be far apart as is the case for the models studied here).
Shirley \cite{shirley} extended the family of one-dimensional lattice models for which the decorrelation estimate may be proved to include alloy-type models with correlated random variables, hopping models, and certain one-dimensional quantum graphs.

%%%%%%%%%%%%%%%%%%%%%%%%%%%%%%%%%%%%%%%%%%%%%%%%%%%%%%%%%%%%%%%%%%%%%%%%%%%%%%%%%%%%%%%%%%%%%
\subsection{Asymptotic independence and decorrelation estimates}\label{subsec:decorrelation1}

The main result is the asymptotic independence of random variables associated with the local eigenvalue statistics
centered at two distinct energies $E$ and $E^\prime$ satisfying $|E-E^\prime | > 4d$.

 We note that in one-dimension there are stronger results and the condition $|E-E^\prime | > 4d$ is not needed. Our results are inspired by the work of Klopp \cite{klopp} for the Anderson models on $\Z^d$ and of Shirley \cite{shirley} for related models on $\Z^d$.
The condition $|E - E^\prime | > 4d$ requires that the two energies be fairly far apart. For example, if $\omega_0 \in [-K, K]$ so
that the deterministic spectrum $\Sigma = [ -2d - K, 2d + K]$, the region of complete localization $\Sigma_{\rm CL}$ is near the band edges $\pm ( 2d + K)$. In this case, one can consider $E$ and $E^\prime$ near each of the band edges.
Our main result on asymptotic independence is the following theorem.

\begin{thm}\label{thm:decorrelation-lattice1}
Let $E,E' \in \Sigma_{\rm CL}$ be two distinct energies with $|E - E'| > 4d$. Let $\xi_{\omega, E}$, respectively, $\xi_{\omega, E'}$,  be a limit point of the local eigenvalue statistics centered at $E$, respectively, at $E'$. Then these two processes are independent. For any bounded intervals $I,J \in \mathcal{B}(\R)$, the random variables $\xi_{\omega, E}(I)$ and $\xi_{\omega, E'}(J)$ are independent random variables distributed according to a compound Poisson process.
\end{thm}

We refer to \cite{germinet-klopp} for a description of the region of complete localization $\Sigma_{\rm CL}$. For information on L\'evy processes, we refer to the books by Applebaum \cite{applebaum} and by Bertoin \cite{bertoin}.
Theorem \ref{thm:decorrelation-lattice1} follows (see section \ref{sec:proof1}) from the following decorrelation estimate. We assume that $L > 0$ is a positive integer, and that $\ell := [ L^\alpha]$ is the greatest integer less than $L^\alpha$ for an exponent $0 < \alpha < 1$. For polymer type models, we assume that $m_k$ divides $L$ and $\ell$.

\begin{pr}\label{prop:decorrelation-prop1}
We choose positive numbers $(\alpha, \beta)$ satisfying \eqref{eq:constraints1} and length scales $L$ and $\ell := [L^\alpha]$ as described above. For a pair of energies $E, E^\prime \in \Sigma_{\rm CL}$, the region of complete localization, with $|E - E^\prime| > 4d$, and bounded intervals $I,J \subset \R$, we define $I_L(E) := L^{-d}I + E$ and $J_L (E^\prime) := L^{-d} J + E^\prime$ as two scaled energy intervals centered at $E$ and $E^\prime$, respectively. We then have
\beq\label{eq:prob1}
\P \{ ({\rm Tr} E_{H_{\omega,\ell}} (I_L (E)) \geq 1 )  \cap  ({\rm Tr} E_{H_{\omega,\ell}} (J_L (E^\prime)) \geq 1 ) \} \leq \frac{C_0}{L^{2d(1 - 2 \beta - \alpha)}} .
\eeq
\end{pr}

The extended Minami estimate \cite{hislop-krishna1} implies that we need to estimate:
\beq\label{eq:prob2}
\P \{ ({\rm Tr} E_{H_{\omega,\ell}} (I_L (E)) \leq m_k )  \cap  ({\rm Tr} E_{H_{\omega,\ell}} (J_L(E^\prime)) \leq m_k ) \}
\eeq
In fact, we consider the more general estimate:
\beq\label{eq:prob3}
\P \{ ({\rm Tr} E_{H_{\omega,L}} (I_L(E)) = k_1 )  \cap  ({\rm Tr} E_{H_{\omega,L}} (J_L(E^\prime)) = k_2 ) \},
\eeq
where $k_1, k_2$ are positive integers independent of $L$.

We allow that there may be several eigenvalues in $I_L(E)$ and $J_L (E^\prime)$ with nontrivial multiplicities.
To deal with this, we introduce the mean trace of the eigenvalues $E_j(\omega)$ of $H_{\omega, \ell}$ in the interval $I_L(E)$:
\beq\label{eq:defn-weight-ave1}
%\mathcal{T}(\omega) := \frac{1}{k_1} \sum_{j=1}^{k_1} E_j(\omega) ~~\mbox{and} ~~\mathcal{T}(\omega)^\prime := \frac{1}{k_2} \sum_{j=1}^{k_2} E_j^\prime(\omega),
\mathcal{T}(\omega) :=
\frac{Tr\big(H_{\omega,\ell} E_{H_{\omega,\ell}}(I_L(E))\big)}{Tr\big(E_{H_{\omega,\ell}}(I_L(E))\big)} = \frac{1}{k_1} \sum_{j=1}^{k_1} E_j(\omega),
\eeq
where $k_1 := Tr\big(E_{H_{\omega,\ell}}(I_L(E))\big)$ is the number of eigenvalues, including multiplicity,
of $H_{\omega, \ell}$ in $I_L(E)$.
Similarly, we define
% ~~\mbox{and} ,
\beq\label{eq:defn-weight-ave2}
\mathcal{T}^\prime(\omega) := \frac{Tr\big(H_{\omega,\ell} E_{H_{\omega,\ell}}(J_L(E^\prime))\big)}{Tr\big(E_{H_{\omega,\ell}}(J_L(E^\prime))\big)}
%\mathcal{T}(\omega)^\prime :
= \frac{1}{k_2} \sum_{j=1}^{k_2} E_j(\omega).
\eeq
We will show in section \ref{sec:sum-estimates1} that these weighted sums behave like an effective
eigenvalues in each scaled interval.

As another application of the extended Minami estimate, we prove that the multiplicity of eigenvalues in $\Sigma_{\rm CL}$ is at most the multiplicity of the perturbations $m_k$ in dimensions $d \geq 1$. The proof of this fact follows the argument of Klein and Molchanov \cite{klein-molchanov}. For $d=1$, Shirley \cite{shirley} proved that the usual Minami estimate holds for the dimer model so the eigenvalues are almost surely simple.

%The proof of this fact follows the argument of Klein and Molchanov \cite{klein-molchanov}.

%%%%%%%%%%%%%%%%%%%%%%%%%%%%%%%%%%%%%%%%%

\subsection{Contents}\label{subsec:contents}

We present properties of the weighted average of eigenvalues in section \ref{sec:sum-estimates1}, including gradients and Hessian estimates. The proof of the main technical result, Proposition \ref{prop:decorrelation-prop1}, is presented in section \ref{sec:proof-prop1}. The proof of asymptotic independence is given in section \ref{sec:proof1}. We show in section \ref{sec:multiplicity1} that the argument of Klein-Molchanov \cite{klein-molchanov} applies to higher rank perturbations and implies that the multiplicity of eigenvalues in $\Sigma_{\rm CL}$ is at most $m_k$, the uniform rank of the perturbations.

%%%%%%%%%%%%%%%%%%%%%%%%%%%%%%%%%%%%%%%%%%%%%%%%%%%%%%%%%%%%%%%%%%%%%%%%%%%%%%%%%%%%%

\section{Estimates on weighted sums of eigenvalues}\label{sec:sum-estimates1}

In this section, we present some technical results on weighted sums of eigenvalues of $H_{\omega, \ell}$
defined in \eqref{eq:defn-weight-ave1}-\eqref{eq:defn-weight-ave2}.
These are used in section \ref{sec:proof1} to prove the main technical result \eqref{eq:prob1}.

%:
%\beq\label{eq:prob1-2}
%\P \{ ({\rm Tr} E_{H_\ell} (I_L (E)) \geq 1 )  \cap  ({\rm Tr} E_{H_\ell} (J_L (E^\prime)) \geq 1 ) \} \leq \frac{C_0}{L^{2d(1 - 2\alpha)}} ,
%\eeq
%where $I_L(E) := L^{-d} I + E$ and $J_L(E^\prime) := L^{-d} I + E^\prime$. In our setting, these intervals may contain $k_1$, respectively $k_2$, not-necessarily distinct
%eigenvalues.
%To deal with this, we work with the weighted sums of eigenvalues
%introduced in \eqref{eq:defn-weight-ave1}.
%\beq\label{eq:prob-main1}
%\P \{ ({\rm Tr} E_{H_{{\ell}}} ({J}_L) \geq k_1 )  \cap
% ({\rm Tr} E_{H_{\tilde{\ell}}} (\tilde{J}_L^\prime) \geq k_2 ) \},
%\eeq
%described above. Recall that $\tilde{\ell} \sim \log L$ and $|\tilde{J}_L | = 4 L^{-1}$.

%%%%%%%%%%%%%%%%%%%%%%%%%%%%%%%%%%%%%%%%%%%%%%%%%%%%%%%%%%%%%%%%%%%%%%%
\subsection{Properties of the weighted trace}\label{subsec:weight-trace1}

When the total number of eigenvalues of $H_{\omega,\ell}$ in $J_L(E) := L^{-d}I + E$ is $k_1$, we get,
\beq\label{defn:trace3}
\mathcal{T}(\omega) := \mathcal{T}_\ell(E,k_1) := \mathcal{T}_\ell (E, k_1, \omega)
% := \mathcal{T}(\omega)
= \frac{1}{k_1} \sum_{j=1}^{k_1} E_j(\omega),
\eeq
for eigenvalues $E_j(\omega) \in {J}_L (E)$.
Properties (1)-(3) below are valid
for the similar expression obtained by replacing $k_1$ with $k_2$, the interval $I$ with $J$, and the energy $E$ with $E^\prime$. We will write
$$
% replacing $E$,  properties (1)-(3) below are valid for
\mathcal{T}^\prime(\omega) := \mathcal{T}_\ell(E^\prime, k_2) := \mathcal{T}(E^\prime, k_2, \omega).
 %:= \mathcal{T}^\prime(\omega).
$$

The weighted eigenvalue average behaves like an effective eigenvalue in the following sense:
\begin{enumerate}
\item $\mathcal{T}_\ell (E, k_1, \omega) \in {J}_L (E)$,
 so the weighted average of the eigenvalue cluster in ${J}_L(E)$ behaves as an eigenvalue in ${J}_L (E)$.

\item Let $E_j (\omega) \in {J}_L (E)$ be an eigenvalue of multiplicity $m_j$. Then a derivative may be computed as follows. Let $\varphi_{j,i}$, for $i=1, \ldots, m_j$, be an orthonormal basis of the eigenspace for $E_j (\omega)$. Then,
\beq\label{eq:deg-ev1}
0 = \frac{\partial}{\partial \omega_s} \sum_{i=1}^{m_j} \langle \varphi_{j,i}, (H_{\omega, \ell} - E_j(\omega)) \varphi_{j,i} \rangle,
\eeq
so we obtain
\beq\label{eq:deg-ev2}
\frac{\partial E_j(\omega)}{\partial \omega_s} = \frac{1}{m_j} \sum_{i=1}^{m_j} \| P_s \varphi_{j,i} \|^2,
\eeq
where $P_s$ is the projector associated with the random variable $\omega_s$.

\item Suppose there are $\hat{k}_1$ distinct eigenvalues in $J_L (E)$ each with multiplicity $m_j$ so $\sum_{j=1}^{\hat{k}_1} m_j = k_1$.
Then, we have
\beq\label{eq:deg-ev3}
\frac{\partial \mathcal{T}_\ell (E, k_1, \omega)}{\partial \omega_s} =
\frac{1}{k_j} \sum_{j=1}^{\hat{k}_1} \sum_{i=1}^{m_j} \| P_s \varphi_{j,i} \|^2 \geq 0.
\eeq
This shows that $\mathcal{T}_\ell(E, k_1, \omega)$ is non-decreasing as a function of $\omega_s$.

\item It follows from \eqref{eq:deg-ev3} that the $\omega$-gradient of the weighted trace is normalized:
$\| \nabla_\omega \mathcal{T}(\omega) \|_{\ell^1} = 1$.

\end{enumerate}

\begin{remark}\label{remark:t-separation}
It follows from property (1) above and the fact that the intervals $I_L(E)$ and $J_L(E)$ are $\mathcal{O}(L^{-d})$, that if $|E - E^\prime| > 4d$,
then $|\mathcal{T}(\omega) - \mathcal{T}(\omega^\prime)| > 4d - cL^{-d}$, for some $c > 0$.
We will use this result below.
\end{remark}

%\beq\label{eq:variation1}
%\omega \cdot \nabla_\omega ( \mathcal{T}(\omega) - \mathcal{T}^\prime (\omega)) = \frac{1}{k} \sum_{i=1}^k \omega \cdot \nabla_\omega
% (E_i(\omega) - E_i^\prime (\omega) )
% \eeq

%%%%%%%%%%%%%%%%%%%%%%%%%%%%%%%%%%%%%%%%%%%%%%%%%%%%%%%%%%%%%%%%%%%%%%%%%%%%%%

\subsection{Variational formulae}\label{subsec:variation1}

We can estimate the variation of the mean trace with respect to the random variables as follows. The $\omega$-directional derivative is
 \bea\label{eq:variation2}
 \omega \cdot \nabla_\omega (\mathcal{T}(\omega) - \mathcal{T}^\prime (\omega) ) &=& \frac{1}{k_1} \sum_{i=1}^{k_1}  \omega \cdot \nabla_\omega E_i (\omega) -
 \frac{1}{k_2} \sum_{j=1}^{k_2} \omega \cdot \nabla_\omega  E_j (\omega) \nonumber \\
  &=& \mathcal{T}(\omega) - \mathcal{T}^\prime (\omega) - \frac{1}{k_1} \sum_{i=1}^{k_1} \langle \varphi_i, (- \Delta) \varphi_i \rangle
    \nonumber \\
    & & + \frac{1}{k_2} \sum_{j=1}^{k_2} \langle \varphi_j, (- \Delta) \varphi_j \rangle .
  \eea

On the lattice, the absolute value of each sum involving the Laplacian may be bounded above by $2 d$. If we assume that
$$
| \mathcal{T}(\omega) - \mathcal{T}^\prime(\omega) | \geq \Delta E
$$
then we obtain from \eqref{eq:variation2},
\bea\label{eq:variation3}
\Delta E - 4d & \leq & | \mathcal{T}(\omega) - \mathcal{T}^\prime(\omega) | - 4d \nonumber \\
 & \leq & |\omega \cdot \nabla_\omega (\mathcal{T}(\omega) - \mathcal{T}^\prime (\omega) )|.
\eea
As the number of components of $\omega$ is bounded by $\ell^d$ and $|\omega_j | \leq K$, it follows by Cauchy-Schwartz inequality that
\beq\label{eq:variation4}
\| \nabla_\omega (\mathcal{T}(\omega) - \mathcal{T}^\prime (\omega) )\|_2 \geq \frac{\Delta E - 4d}{K} \frac{1}{(2\ell + 1)^{d/2}}.
\eeq
We also obtain an $\ell^1$ lower bound:
\beq\label{eq:variation5}
\| \nabla_\omega (\mathcal{T}(\omega) - \mathcal{T}^\prime (\omega) )\|_1 \geq \frac{\Delta E - 4d}{K} .
\eeq

%%%%%%%%%%%%%%%%%%%%%%%%%%%%%%%%%%%%%%%%%%%%%%%%%%%%%%%%%%%%%%%%%%%%%%%%%%%%%%%%%%%%%%%%%%%%%%%%%

\subsection{Hessian estimate}\label{subsec:hessian1}

The Hessian of $\mathcal{T}(\omega)$ has $ij^{\rm th}$ matrix elements given by
\beq\label{eq:hessian1}
{\rm Hess}(\mathcal{T})_{ij} = \frac{1}{k_1} \sum_{m=1}^{k_1} \frac{\partial^2}{\partial \omega_i \partial \omega_j} E_m (\omega) .
\eeq
It is convenient to compute this using trace notation. Let $P_E$ denote the spectral projection onto the eigenspace of $H_{\omega, \ell}$
corresponding to the eigenvalues $E_m(\omega)$ in $J_L(E)$. Let $\gamma_E$ be a simple closed contour containing only these eigenvalues of $H_{\omega, \ell}$ with a counter-clockwise orientation. Since the weighted mean of the eigenvalues may be expressed as
$$
\mathcal{T}(\omega) = \frac{1}{k_1} {\rm Tr} H_{\omega,\ell} P_E,
$$
and the projection has the representation
$$
P_E = \frac{1}{2 \pi i} \int_{\gamma_E} R(z) ~dz, ~~~ R(z) := (H_{\omega, \ell} - z)^{-1},
$$
it follows that
\beq\label{eq:tr-deriv1}
\frac{\partial}{\partial \omega_j} \mathcal{T}(\omega) = \frac{-1}{2 \pi i k_1} \int_{\gamma_E} {\rm Tr} \{ R(z) P_j R(z) \} ~z dz ,
\eeq
where $P_j$ is the finite-rank projector associated with site $j$ or block $j$, depending on the model.
Computing the second derivative, the matrix elements of the Hessian of $\mathcal{T}(\omega)$ are
\bea\label{eq:hessian2}
{\rm Hess}(\mathcal{T})_{ij} & =  & \frac{1}{2 \pi i k} \int_{\gamma_E} {\rm Tr} \{ R(z)P_i R(z)P_j R(z)  \nonumber \\
 & &  +   R(z) P_j R(z) P_j R(z) \} ~ z dz
 %\nonumber \\
 %& \leq & \frac{-1}{2 \pi i} \int_{\gamma_E} | \langle \delta_i, R(z) \delta_j \rangle |^2 ~dz
\eea
This formula will provide the equivalent of Lemma 2.3 \cite{klopp}, for both $Hess({\mathcal T})_{ij}, Hess(\mathcal{T}^\prime)_{ij}$.

%since the two powers of the resolvent multiplied by the length of the contour $\gamma_E$ give an upper bound on the %Hessian of the inverse of the distance from $\mathcal{T}$ to the spectrum of $H_\ell$ minus the eigenvalues in %$J_L$.

\begin{lemma}\label{lemma:hessian1}
The Hessian of the weighted average $\mathcal{T}(\omega)$ of the eigenvalues of $H_\ell (\omega)$ in an interval of order $L^{-d}$ satisfies the bound:
\bea\label{eq:hessian3}
\| {\rm Hess}(\mathcal{T})_{ij} \|_{\ell^\infty \rightarrow \ell^1} & \leq &  |\gamma_E|^2 \sup_{z \in \gamma_E} \| P_i R(z)^2 P_j \|_1 \| P_j R(z) P_1 \|_1 \nonumber \\
 & \leq & C \frac{L^{-2d}}{ ({\rm dist} ( \gamma_E, \sigma(H_{\omega, \ell}) ))^3}.
 \eea
Since the Wegner estimate insures that ${\rm dist} ( \gamma_E, \sigma(H_{\omega, \ell}) ) \sim \ell^{-d}$ with probability greater than $1 - C_W (\ell / L)^d$, we obtain
\beq\label{eq:hessian4}
 \| {\rm Hess}(\mathcal{T})_{ij} \|_{\ell^\infty \rightarrow \ell^1}  \leq C L^{-2d} \ell^{3d} \leq C L^{3 d \alpha - 2d} ,
 \eeq
 so if $0 < \alpha < 2/3$, the Hessian is vanishes as $L \rightarrow \infty $.  The above statements are also valid for $\mathcal{T}^\prime (\omega)$.
 % also. A similar estimate holds for $\mathcal{T}^\prime (\omega)$.
\end{lemma}

\section{Proof of Proposition \ref{prop:decorrelation-prop1}}\label{sec:proof-prop1}

%In this section, we give the proof of Theorem \ref{thm:decorrelation-lattice1}.
In this section, we prove the technical result, Proposition \ref{prop:decorrelation-prop1}.
We let
$X_\ell(I_L(E)) := {\rm Tr} E_{H_{\omega,\ell}} (I_L(E))$,
$X_\ell(J_L(E^\prime)) := {\rm Tr} E_{H_{\omega,\ell}} (J_L(E^\prime))$.
Then, we show
\beq\label{eq:same-number1}
\P \{ ( X_\ell(I_L(E)) \geq 1 ) \cap ( X_\ell(J_L(E')) \geq 1) \} \leq C_0  \frac{1}{L^{2d(1 - 2 \beta - \alpha)}} ,
\eeq
for positive numbers $(\alpha, \beta)$ satisfying \eqref{eq:constraints1}.

%In this section, we prove the main technical result:
%\beq\label{eq:prob-main1}
%\P_1(\tilde{\ell}) := \P \{ ({\rm Tr} E_{H_{\tilde{\ell}}} (\tilde{J}_L) \geq k_1 )  \cap
% ({\rm Tr} E_{H_{\tilde{\ell}}} (\tilde{J}_L^\prime) \geq k_2 ) \},
%\eeq
%described above. Recall that $\tilde{\ell} \sim \log L$ and $|\tilde{J}_L | = 4 L^{-1}$.
%

%%%%%%%%%%%%%%%%%%%%%%%%%%%%%%%%%%%%%%%%%%%%%%%%%%%%%%%%%%%%%%%%%%%%%%%%%%%%%%%%%%%%%%%%%%%%%%%%%%%

\subsection{Reduction via the extended Minami estimate}\label{subsec:reduction1}

%We let $X_\ell(E) := {\rm Tr} E_{J_L} (H_\ell)$.
Let $\chi_A(\omega)$ be the characteristic function on the subset $A \subset \Omega$. In this section, we write $J_L(E) := L^{-d} J + E$ since we are dealing with one interval.
We use an extended Minami estimate of the form
\beq\label{eq:ext-minami1}
\E \{ \chi_{ \{ \omega ~|~  X_\ell(J_L(E)) \geq m_k +1 \} } X_\ell(J_L(E))( X_\ell(J_L(E)) - m_k ) \geq 1 \}  \nonumber \\
  \leq  C_M \left( \frac{\ell}{L} \right)^{2d} ,
%\P \{ X_\ell (E) \geq m_k \} \leq C_M \left( \frac{\ell}{L} \right)^{2d} .
\eeq
as follows from \cite{hislop-krishna1}.

\begin{lemma}\label{lemma:minami-consequence1}
Under the condition that the projectors have uniform dimension $m_k \geq 1$, we have
\beq\label{eq:ext-minami2}
\P \{ X_\ell(J_L(E)) > m_k \} \leq C_M \left( \frac{\ell}{L}\right)^{2d} .
\eeq
\end{lemma}

\begin{proof}
Recalling that $X_\ell(J_L(E)) \in \{ 0 \} \cup \N$, we have
\bea\label{eq:minami-conseq1}
\lefteqn{ \P \{ X_\ell(J_L(E)) > m_k \} } \nonumber \\
 & \leq & \P \{ X_\ell(J_L(E)) - m_k \geq 1 \} \nonumber \\
 & = &  \P \{ X_\ell(J_L(E)) ( X_\ell(J_L(E)) - m_k ) \geq 1 \}  \nonumber \\
 & = &  \P \{ \chi_{ \{ \omega ~|~  X_\ell(J_L(E)) \geq m_k +1 \} } X_\ell(J_L(E))( X_\ell(J_L(E)) - m_k ) \geq 1 \}  \nonumber \\
 & \leq & \E \{ \chi_{ \{ \omega ~|~  X_\ell(J_L(E)) \geq m_k +1 \} } X_\ell(J_L(E))( X_\ell(J_L(E)) - m_k ) \geq 1 \}  \nonumber \\
 & \leq & C_M \left( \frac{\ell}{L} \right)^{2d} ,
 \eea
by the extended Minami estimate \cite{hislop-krishna1}.
\end{proof}

%We let $X_\ell(E) := {\rm Tr} E_{J_L} (H_\ell)$. We assume an extended Minami estimate of the form
%\beq\label{eq:ext-minami1}
%\P \{ X_\ell (E) \geq m_k \} \leq C_M \left( \frac{\ell}{L} \right)^{2d} .
%\eeq
%
%In this section, we give the proof that
%\beq\label{eq:same-number1}
%\P \{ ( X_\ell(E) \geq 1 ) \cap ( X_\ell(E') \geq 1) \} \leq C_0 \left( \frac{\ell}{L} \right)^{2d} .
%\eeq

%%%%%%%%%%%%%%%%%%%%%%%%%%%%%%%%%%%%%%%%%%%%%%%%%%%%%%%%%%%%%%%%%%%%%%%%%%%%%%%%%%%%%%%%%%%%%%%%%%%%%%%%%%%%%%%%%%%
\subsection{Estimates on the joint probability}\label{subsec:joint-prob1}

We return to considering two scaled intervals $I_L(E)$ and $J_L(E^\prime)$, with $E \neq E^\prime$.
Because of \eqref{eq:ext-minami1}, we have
\bea\label{eq:same-number2}
\lefteqn{ \P \{ ( X_\ell(I_L(E)) \geq 1 ) \cap ( X_\ell(J_L(E')) \geq 1) \} } \nonumber \\
 & \leq & \P \{ ( X_\ell(I_L(E)) \geq m_k +1 ) \cap ( X_\ell(J_L(E^\prime)) \geq m_k +1 ) \} \nonumber \\
 & & +  \P \{ ( X_\ell(I_L(E)) \leq m_k ) \cap ( X_\ell(J_L(E^\prime)) \geq m_k +1 ) \}
 \nonumber \\
 & & +  \P \{ ( X_\ell(I_L(E)) \leq m_k + 1 ) \cap ( X_\ell(J_L(E^\prime)) \geq m_k ) \}
 \nonumber \\
 & & + \P \{ ( X_\ell(I_L(E)) \leq m_k  ) \cap ( X_\ell(J_L(E^\prime)) \leq m_k ) \}
 \nonumber \\
 & \leq & \P \{ ( X_\ell(I_L(E)) \leq m_k ) \cap ( X_\ell(J_L(E^\prime)) \leq m_k  ) \} \nonumber \\
 & &
 + C_0 \left( \frac{\ell}{L} \right)^{2d} .
\eea
The probability on the last line of \eqref{eq:same-number2} may be bounded above by
\bea\label{eq:reduction1}
\lefteqn{ \P \{ ( X_\ell(I_L(E)) \leq m_k ) \cap ( X_\ell(J_L(E^\prime)) \leq m_k  ) \} } & & \nonumber \\
 & \leq &
\sum_{k_1,k_2}^{m_k} \P \{ ( X_\ell(I_L(E))  = k_1 ) \cap ( X_\ell(J_L(E^\prime)) = k_2  ) \}.
\eea
Since $m_k$ is independent of $L$, it suffices to estimate
\beq\label{eq:same-number3}
\P \{ ( X_\ell(I_L(E))  = k_1 ) \cap ( X_\ell(J_L(E^\prime)) = k_2  ) \}.
\eeq

The proof of the next key Proposition \ref{prop:vanishing-prob1} follows the ideas in \cite{klopp}.

\begin{pr}\label{prop:vanishing-prob1}
For $k_1, k_2 = 1, \ldots, m_k$ and positive numbers $(\alpha, \beta)$ satisfying \eqref{eq:constraints1}, we have
\beq\label{eq:vanishing-prob1}
\P \{ ( X_\ell(I_L(E))  = k_1 ) \cap ( X_\ell(J_L(E^\prime)) = k_2  ) \} \leq
C \left( \frac{K}{\Delta E - 4d} \right)^4 L^{-2d (1 - 2 \beta - \alpha)}.
\eeq
\end{pr}

%Let $\{ E_j(\ell, \omega) \}$ be the eigenvalues of $H_\ell$, listed in increasing order including multiplicity.
%For a given configuration $\omega$, the number of eigenvalues of $H_\ell$ in $J_L(E)$ is $X_\ell(E) = {\rm Tr} E_{J_L(E)}(H_\ell)$.
%For the event $X_\ell(J_L(E))  = k_1$, we recall that $\mathcal{T}_\ell (E,k_1)$ is the weighted average of the $k_1$
%eigenvalues of $H_\ell$ in $J_L(E)$:
%\beq\label{eq:weight-ave1}
%\mathcal{T}_\ell (E, k_1) := \frac{1}{k_1} \sum_{j=1}^{k_1} E_j(\ell, \omega).
%\eeq

%The proof of Proposition \ref{prop:vanishing-prob1} follows the ideas in \cite{klopp}.

\begin{proof}
1. We begin with some observation concerning the eigenvalue averages.
We let $\Omega_0(\ell, k_1, k_2)$ denote the event
\beq\label{eq:joint-prob1}
\Omega_0(\ell, k_1, k_2) := \{ \omega ~|~  (X_\ell(I_L(E)) = k_1 ) \cap ( X_\ell(J_L(E')) = k_2) \} \cap \Omega_{W, \ell},
\eeq
for $k_1, k_2 = 1, \ldots, m_k$. The set $\Omega_{W, \ell}$ is the set of
$\omega$ for which the eigenvalue spacing for $H_{\omega, \ell}$ in the interval $I_L(E)$ or $J_L(E^\prime)$ is $\mathcal{O}(\ell^{-d})$. By the Wegner estimate, the probability of this set is at least $1 - C_W ( \ell / L)^{-d}$, as discussed in Lemma \ref{lemma:hessian1}. We define the subset $\Delta \subset \Lambda_\ell \times \Lambda_\ell$ by $\Delta := \{ (i,i) ~|~ i \in \Lambda_\ell \}$. For each pair of sites $(i,j) \in \Lambda_\ell \times \Lambda_\ell \backslash \Delta$, the Jacobian determinant of the mapping $\varphi: (\omega_i, \omega_j) \rightarrow
( \mathcal{T}_\ell (E, k_1), \mathcal{T}_\ell (E^\prime, k_2))$, given by:
\beq\label{eq:jacobian1}
     J_{ij}( \mathcal{T}_\ell (E,k_1), \mathcal{T}_\ell(E^\prime, k_2) ) :=     \left|    \begin{array}{cc}
                   \partial_{\omega_i}\mathcal{T}_\ell (E, k_1) & \partial_{\omega_j} \mathcal{T}_\ell (E, k_1) \\
                   \partial_{\omega_i} \mathcal{T}_\ell (E^\prime, k_2) & \partial_{\omega_j} \mathcal{T}_\ell (E^\prime, k_2)
                    \end{array}   \right| .
\eeq
As we will show in section \ref{subsec:diffeom1}, the condition $J_{ij}( \mathcal{T}_\ell (E,k_1), \mathcal{T}_\ell(E^\prime, k_2) ) \geq \lambda (L) > 0$ implies that the average of the eigenvalues in $I_L(E)$ and $J_L(E^\prime)$ effectively vary independently with respect to any pair of independent random variables $(\omega_i, \omega_j)$, for $i \neq j$. We define the following events for pairs $(i,j) in \Lambda_\ell \times \Lambda_\ell \backslash \Delta$:
\beq\label{eq:pair-event1}
\Omega_0^{i,j} (\ell, k_1, k_2) := \Omega_0(\ell, k_1, k_2) \cap \{ \omega ~|~ J_{ij}( \mathcal{T}_\ell (E,k_1), \mathcal{T}_\ell(E^\prime, k_2) ) \geq \lambda (L) \},
\eeq
where $\lambda (L) > 0$ is given by
\beq\label{eq:defn-lambda1}
\lambda(L) := ( \Delta E - 4d) K^{-1} L^{- \beta d} ,
\eeq
where the exponent $\beta > 0$ satisfies
\beq\label{eq:constraints1}
0 < \beta < \frac{1}{2}, ~~0 < \alpha + 4 \beta < 1, ~~ 0 < \alpha < \frac{2}{5} \beta .
\eeq
For example, we may take $\beta = \frac{1}{8}$ and $\alpha < \frac{1}{20}$.

\noindent
2. We next compute $\P \{ \Omega_0^{i,j} (\ell, k_1, k_2) \}$. Following Klopp \cite[pg.\ 242]{klopp}, we prove in section \ref{subsec:diffeom1} that the positivity of the Jacobian determinant insures that the map $\varphi$, restricted to a certain domain, is a diffeomorphism. In particular, for any pair
$(i,j) \in \Lambda_\ell \times \Lambda_\ell \backslash \Delta$, if $(\omega_i^0 , \omega_j^0 , \omega_{ij}^\perp) \in \Omega_0^{i,j}(\ell,k_1,k_2)$, then it follows from Lemma \ref{lemma:diffeom1} if $\| (\omega_i^0, \omega_j^0) - (\omega_i, \omega_j) \| > L^{-d} \lambda^{-2}(L)$, consequently
one has $(\mathcal{T}_\ell(E, k_1, \omega), \mathcal{T}_\ell (E^\prime, k_2, \omega)) \not \in {I}_L(E) \times {J}_L(E^\prime)$.
This would contradict the fact that $\omega \in \Omega_0 (\ell, k_1, k_2)$.
This is key in the following computation:
\bea\label{eq:independent-est1}
& & \P \{ \Omega_0^{i,j} (\ell, k_1, k_2) \}\nonumber \\ & = & \E_{\omega_{ij}^\perp} \left\{ \int_{\R^2} \chi_{\Omega_0^{i,j} (\ell, k_1, k_2)} ( \omega_i, \omega_j, \omega_{ij}^\perp ) g(\omega_i) g(\omega_j) ~d\omega_i ~d \omega_j \right\} \nonumber \\
 & \leq & \E_{\omega_{ij}^\perp} \left\{ \int_{\R^2} \chi_{ \{ \| (\omega_i, \omega_j) - ( \omega_i^0, \omega_j^0) \|_\infty
 \leq L^{-d} \lambda^{-2} \} } ( \omega_i, \omega_j, \omega_{ij}^\perp ) g(\omega_i) g(\omega_j) ~d\omega_i ~d \omega_j \right\} \nonumber \\
 & \leq & C L^{-2d} \lambda^{-4}(L) .
 \eea
 %Let $\Delta := \{ (i,i) ~|~ i \in \Lambda_\ell \}$.

\noindent
3. We next bound $\P \{ \Omega_0 (\ell, k_1, k_2) \}$ in terms of $\P \{ \Omega_0^{i,j} (\ell, k_1, k_2) \}$ using
\cite[Lemma 2.5]{klopp}. This lemma states that for $(u,v ) \in (\R^+)^{2n}$ normalized so that $\| u \|_1 = \| v \|_1 = 1$, we have
\beq\label{eq:jac-lb1}
\max_{j \neq k} \left| \begin{array}{cc}
                            u_j & u_k \\
                            v_j & v_k
                            \end{array} \right|^2 \geq \frac{1}{4 n^5} \| u-v \|_1^2.
\eeq
Applying this with $n = (2 \ell + 1)^d$, and  $u = \nabla_\omega \mathcal{T}(\omega)$ and $v = \nabla_{\omega} \mathcal{T}^\prime (\omega)$,
and recalling the positivity \eqref{eq:deg-ev3} in point (3) and the normalization in point (4) of section \ref{subsec:weight-trace1}, we obtain from \eqref{eq:jac-lb1} and \eqref{eq:variation5}:
\bea\label{eq:jac-lb2}
\max_{i \neq j \in \Lambda_\ell} J_{ij} ( \mathcal{T}_\ell (E), \mathcal{T}_\ell (E^\prime))^2  &\geq & \left( \frac{2^3}{ \ell^{5d}} \right)
  \| \nabla_\omega ( \mathcal{T}_\ell (E) - \mathcal{T}_\ell (E^\prime)) \|_1^2 \nonumber \\
 & \geq &
\left( \frac{\Delta E - 4d}{K} \right)^2 \left( \frac{2^3}{ \ell^{5d}} \right) .
\eea
We partition the probability space as $\{ \omega ~|~ J_{ij} \geq \lambda(L) ~{\rm some} ~(i,j) \in \Lambda_\ell \times \Lambda_\ell \backslash \Delta \} \cup \{ \omega ~|~ J_{ij} < \lambda(L) ~\forall ~(i,j) \in \Lambda_\ell \times \Lambda_\ell \backslash \Delta \}$, where we write $J_{ij}$ for the Jacobian $J_{ij} ( \mathcal{T}_\ell (E), \mathcal{T}_\ell (E^\prime))$.
Suppose that the second event $\{ \omega ~|~ J_{ij} < \lambda(L) ~\forall ~(i,j) \in \Lambda_\ell \times \Lambda_\ell \backslash \Delta \}$ occurs, so that from \eqref{eq:jac-lb2}, we have:
\bea\label{eq:jac-lb3}
\lefteqn{ \lambda(L)^2 = \left( \frac{C_0}{L^{\beta d}} \right)^2 \geq  \max_{i \neq j \in \Lambda_\ell} J_{ij} ( \mathcal{T}_\ell (E), \mathcal{T}_\ell (E^\prime))^2 } \nonumber \\
  &\geq & \left( \frac{2^3}{ \ell^{5d}} \right)
  \| \nabla_\omega ( \mathcal{T}_\ell (E) - \mathcal{T}_\ell (E^\prime)) \|_1^2 .
\eea
This implies that
\beq\label{eq:gradients1}
    \| \nabla_\omega ( \mathcal{T}_\ell (E) - \mathcal{T}_\ell (E^\prime)) \|_1 \leq C_1 L^{-d(\beta - 5 \alpha / 2 )} .
\eeq
So, provided $0 < \alpha < \frac{2}{5} \beta$, we find that the bound \eqref{eq:gradients1} implies that the $\nabla_\omega \mathcal{T}_\ell (E)$ is almost collinear with $\nabla_\omega \mathcal{T}_\ell (E^\prime)$. This contradicts the lower bound \eqref{eq:variation5} as long as $\Delta E > 0$. Consequently, the probability of the second event is zero.

\noindent
4. It follows from this and the partition of the probability space that
\bea\label{eq:independent-est2}
 \P \{ \Omega_0 (\ell, k_1, k_2) \} & \leq & \sum_{(i,j) \in \Lambda_\ell \times \Lambda_\ell \backslash \Delta} \P \{ \Omega_0^{i,j} (\ell, k_1, k_2) \} \nonumber \\
  & \leq & \ell^{2d} \lambda^{-4}(L) L^{-2d}.
  \eea
We now take $\ell = L^\alpha$ and $\lambda(L) := ( \Delta E - 4d) K^{-1} L^{- \beta d}$, with $(\alpha, \beta)$ satisfying \eqref{eq:constraints1}. With these choices, and the fact that $m_k$ is independent of $L$,
 we obtain the probability
 \beq\label{eq:independent-est3}
  \P \{ \Omega_0 (\ell, k_1, k_2) \} \leq C \left( \frac{K}{\Delta E - 4d} \right)^4 L^{-2d (1 - 2 \beta - \alpha)}.
  \eeq
For choices $\alpha$ and $\beta$ with $0 < \alpha + 2 \beta < 1$, this shows that
$$
\P \{ \Omega_0 (\ell, k_1, k_2) \} ~~{\rm and} ~~ \P \{ (X_\ell(I_L(E)) = k_1 ) \cap ( X_\ell(J_L(E^\prime)) = k_2) \} \rightarrow 0, {\rm as} ~~L \rightarrow 0,
$$
for any integers $k_1, k_2 = 1, \ldots, m_k$.
This proves, up to the proof of the diffeomorphism property of $\varphi$, the main result \eqref{eq:prob1}.
\end{proof}

%%%%%%%%%%%%%%%%%%%%%%%%%%%%%%%%%%%%%%%%%%%%%%%%%%%%%%%%%%%%%%%%%%%%%%%%%%%%%%%%%%%%%%%%%%%
\subsection{Proof of the diffeomorphism property}\label{subsec:diffeom1}

We prove the following lemma on the perturbation of a set of good configurations $(\omega_i^0, \omega_j^0)$.
Let $\Omega_0(\ell, k_1, k_2), k_1, k_2 = 1, \ldots, m_k$ be the set of configurations described in \eqref{eq:joint-prob1}. Similarly, for any pair of sites $(i,j) \in \Lambda_\ell
\times \Lambda_\ell \backslash \Delta$, the Jacobian determinant
$J_{ij}( \mathcal{T}_\ell (E,k_1), \mathcal{T}_\ell(E^\prime, k_2) )$ is defined in equation (\ref{eq:jacobian1}).
%As we will show, the condition $J_{ij}( \mathcal{T}_\ell (E,k_1), \mathcal{T}_\ell(E^\prime, k_2) ) \geq \lambda > %0$, for appropriate $\lambda$ chosen below, implies that the weighted average of the eigenvalues in $I_L(E)$ and %$J_L(E^\prime)$ vary independently with respect to any pair of independent random variables $(\omega_i, \omega_j)$, %for $i \neq j$.
We also defined events $\Omega_0^{i,j} (\ell, k_1, k_2)$, for pairs $(i,j) \in \Lambda_\ell \times \Lambda_\ell \backslash \Delta$, in \eqref{eq:pair-event1}:
%where $\Delta := \{ (i,i) ~|~ i \in \Lambda_\ell \}$
\beq\label{eq:pair-event2}
\Omega_0^{i,j} (\ell, k_1, k_2) := \Omega_0(\ell, k_1, k_2) \cap \{ \omega ~|~ J_{ij}( \mathcal{T}_\ell (E,k_1), \mathcal{T}_\ell(E^\prime, k_2) ) \geq \lambda(L) \},
\eeq
where $\lambda (L) > 0$ has the value
\beq\label{eq:lambda1}
\lambda (L) := \frac{\Delta E - 4d}{K} L^{- d \beta} ,
\eeq
where $\alpha$ and $\beta$ satisfy the constraints in \eqref{eq:constraints1}.
The stability estimate following from the diffeomorphism property of $\varphi$ is given in the following lemma.

\begin{lemma}\cite[Lemma 2.6]{klopp}\label{lemma:diffeom1} %Let $\epsilon := L^{-d} \lambda^{-3}$.
Suppose that $(\omega_i^0, \omega_j^0, \omega_{ij}^\perp) \in \Omega_{0}^{i,j}(\ell, k_1, k_2)$ and $(\alpha, \beta)$
satisfy \eqref{eq:constraints1}.
Then for any pair $(\omega_i, \omega_j) \in \R^2$ with
$$
\| (\omega_i^0, \omega_j^0) - (\omega_i, \omega_j) \| > L^{-d} \lambda^{-2}(L),
$$
one has
\beq\label{eq:config-pert1-1}
(\mathcal{T}_\ell(E, k_1, \omega), \mathcal{T}_\ell (E^\prime, k_2, \omega)) \not \in {I}_L(E) \times {J}_L(E^\prime),
\eeq
where $\lambda (L)$ has the value given in \eqref{eq:lambda1}.
%with probability at least $1 - C_0(\ell /L)^d$.
\end{lemma}

\begin{proof}
1. %From point (3) of properties in section \ref{subsec:weight-trace1}, the map $\omega_j \rightarrow \mathcal{T}(\omega)$ is nondecreasing.
Let us fix $\omega_{ij}^\perp$ so that $(\omega_i^0, \omega_j^0, \omega_{ij}^\perp) \in \Omega_{0}^{i,j}(\ell, k_1, k_2)$.
% We consider $\omega_{ij}^\perp$ fixed so this condition holds.
We consider the square $\mathcal{S}$ in two-dimensional configuration space:
\beq\label{eq:s-defn1}
\mathcal{S} := \{ (\omega_i, \omega_j) ~|~ \| (\omega_i^0, \omega_j^0) - (\omega_i, \omega_j) \| \leq L^{-d} \lambda (L)^{-2} \}
\eeq
and the map $\varphi : \mathcal{S} \rightarrow \R^2$ defined by
$$
\varphi (\omega_i, \omega_j) = ( \mathcal{T}_\ell (E, k_1, \omega), \mathcal{T}_\ell (E^\prime, k_2,\omega)) .
$$
The first goal is to prove that $\varphi$ is an invertible map between $\mathcal{S}$ and its range $\varphi (\mathcal{S})$.

\noindent
2. To prove that $\varphi$ is injective, we
suppose $(\omega_i, \omega_j)$ and $(\omega_i^\prime, \omega_j^\prime)$ both belong to $\mathcal{S}$ and that
$\varphi ( \omega_i, \omega_j) = \varphi (\omega_i^\prime, \omega_j^\prime)$. Let $D_{ij}\varphi$ denote the $2 \times 2$ matrix that is the derivative of $\varphi$ with respect to $(\omega_i, \omega_j)$.
By the Fundamental Theorem of Calculus, the definition of $\mathcal{S}$, and the Hessian estimate \eqref{eq:hessian4}, we have
\bea\label{eq:grad-est1}
\lefteqn{ \| D_{ij} \varphi (\omega_i, \omega_j) - D_{ij} \varphi (\omega_i^\prime, \omega_j^\prime) \| } \nonumber \\
 & \leq  &
( \| {\rm Hess} \mathcal{T}(\omega) \|  + \| {\rm Hess} \mathcal{T}^\prime(\omega) \| ) L^{-d} \lambda (L)^{-2} \leq  C_0 L^{-(4-3 \alpha - 4\beta)d}.
\eea
The exponent is positive if $0 < \frac{3}{4} \alpha + \beta < 1$ that is satisfied
due to \eqref{eq:constraints1}. By a Taylor's expansion and the Hessian estimate \eqref{eq:hessian4},
%estimate \eqref{eq:grad-est1} on the gradient,
we obtain
\beq\label{eq:grad-est2}
\| \varphi (\omega_i, \omega_j) -  \varphi (\omega_i^\prime, \omega_j^\prime) -
D_{ij} \varphi(\omega_i^0, \omega_j^0) \cdot ({\bf{\omega}} - {\bf{\omega}^\prime}) \| \leq C L^{(3 \alpha -2)d}  \| ({\bf{\omega}}^\prime - {\bf{\omega}}) \|^2  .
\eeq
As a consequence, we can bound
the difference
$$
\| \varphi (\omega_i, \omega_j) -  \varphi (\omega_i^\prime, \omega_j^\prime) \|
$$
from below. Recall that the Jacobian determinant of $D_{ij} \varphi (\omega_i^0, \omega_j^0)$ is bounded below by $\lambda (L)$
since $(\omega_i^0, \omega_j^0) \in \mathcal{S}$. For any pair $(\omega_i, \omega_j), (\omega_i^\prime, \omega_j^\prime ) \in \mathcal{S}$,
we have $\|(\omega_i, \omega_j) - (\omega_i^\prime, \omega_j^\prime )\| \leq C L^{-(1 - 2 \beta)d}$. These facts, the Hessian estimate in \eqref{eq:grad-est1}, and the Taylor expansion in \eqref{eq:grad-est2} yield
\bea\label{eq:bijection1}
\| \varphi (\omega_i, \omega_j) -  \varphi (\omega_i^\prime, \omega_j^\prime) \| & \geq & | D_{ij} \varphi(\omega_i^0, \omega_j^0) \cdot ({\bf{\omega}}^\prime - {\bf{\omega}}) | -
C L^{(3 \alpha - 2)d} \| ({\bf{\omega}}^\prime - {\bf{\omega}}) \|^2 \nonumber \\
 & \geq & C L^{- d \beta} \| ({\bf{\omega}}^\prime - {\bf{\omega}}) \| -  C L^{- d(3 - 3 \alpha - 2 \beta)} \| ({\bf{\omega}}^\prime - {\bf{\omega}}) \| \nonumber \\
 & \geq & C_0(L)  \| ({\bf{\omega}}^\prime - {\bf{\omega}}) \| ,
  \eea
where $C_0(L) := C (L^{-d \beta} - L^{- d(3 - 3 \alpha- 2 \beta )}) > 0$ is strictly positive for $0 < \alpha + \beta < 1$.
This proves the injectivity of $\varphi$.

\noindent
3. We next show that $\varphi$ is an analytic diffeomorphism from $\mathcal{S}$ onto its range. Estimate \eqref{eq:grad-est1} implies that the Jacobians are close:
\beq\label{eq:jacobian2}
| {\rm Jac} \varphi ( \omega_i^\prime, \omega_j^\prime) - {\rm Jac} \varphi ( \omega_i^0, \omega_j^0) | \leq C L^{- (3 - 3 \alpha - 2 \beta)d} .
\eeq
Since $( \omega_i^0, \omega_j^0) \in \Omega_{0}^{i,j}(\ell, k_1, k_2)$, we know
that $| {\rm J}_{ij} ( \mathcal{T}(\omega^0) , \mathcal{T}^\prime(\omega^0)) | \geq \lambda (L)$.
This lower bound and \eqref{eq:jacobian2} imply that for all $(\omega_i, \omega_j) \in \mathcal{S}$ we have
\beq\label{eq:jacobian3}
{\rm J}_{ij} ( \mathcal{T}(\omega), \mathcal{T}^\prime (\omega)) \geq C [ L^{-d \beta} - L^{-(3 - 3 \alpha - 2 \beta )d} ] > 0,
\eeq
provided $0 < \alpha +  \beta < 1$.
Consequently, for all $(\omega_i, \omega_j) \in \mathcal{S}$  and $L$ large enough, the Inverse Function Theorem implies that $\varphi$ is an analytic diffeomorphism. Furthermore, the Jacobian of $\varphi^{-1}$ satisfies the bound
\beq\label{eq:inv-jac1}
| {\rm Jac} \varphi^{-1}(\omega_i, \omega_j)| \leq C L^{d \beta}.
\eeq

\noindent
4. To complete the proof of the lemma, we recall that the map $\omega \rightarrow \mathcal{T}_\ell ( E, k_1, \omega)$ is nondecreasing as shown in section \ref{subsec:weight-trace1}. Hence, we can consider $\|(\omega_i, \omega_j) - ( \omega_i^0, \omega_j^0) \|_\infty = L^{-d} \lambda^{-2}(L)$.
Let us suppose, to the contrary, that for some such pair $(\omega_i, \omega_j) \in \R^2$ with
$$
\| (\omega_i^0, \omega_j^0) - (\omega_i, \omega_j) \| = L^{-d} \lambda^{-2}(L) = C L^{- (1 - 2 \beta)d},
$$
one has
\beq\label{eq:config-pert3}
(\mathcal{T}_\ell (E, k_1, \omega), \mathcal{T}_\ell (E^\prime, k_2,\omega))   \in {I}_L (E) \times {J}_L (E^\prime).
\eeq
Then, using the bound \eqref{eq:inv-jac1}, we have
\bea\label{eq:contradict1}
L^{-d} \lambda^{-2} (L)= C L^{- d(1 - 2 \beta)} & < & \| (\omega_i^0, \omega_j^0) - (\omega_i, \omega_j) \| \nonumber \\
 & = & \| \varphi^{-1} ( \mathcal{T}_\ell (E, k_1, \omega), \mathcal{T}_\ell (E^\prime, k_2, \omega))
  - \varphi^{-1} ( E, E^\prime ) \| \nonumber \\
 & \leq & C L^{-d}L^{\beta d} = C L^{-d (1 - \beta)}.
 \eea
 As $L \rightarrow \infty$, we obtain a contradiction since $\beta > 0$.
 \end{proof}

%In the calculations, we take $\lambda = \mathcal{O} (L^{-d/2})$ so the condition of Lemma \ref{lemma:diffeom1} is satisfied.

%%%%%%%%%%%%%%%%%%%%%%%%%%%%%%%%%%%%%%%%%%%%%%%%%%%%%%%%%%%%%%%%%%%%%%%%%%%%%%%%%%%%%%%%%%%%%%%%%%%%%%%%%%%%%%%%%%%%%%%%%%%%%%%%%%%%%%%%%%%%%

\section{Asymptotically independent random variables: Proof of Theorem \ref{thm:decorrelation-lattice1}}\label{sec:proof1}

In this section, we give the proof of Theorem \ref{thm:decorrelation-lattice1}.
To prove that $\xi_E^\omega(I)$ and $\xi_{E^\prime}^\omega (J)$ are independent, we recall that the limit points $\xi_E^\omega$ are the same as
those obtained from a certain uniformly asymptotically negligible array (\cite[Proposition 4.4]{hislop-krishna1}). To obtain this array,
we construct a cover of $\Lambda_L$ by non-overlapping cubes of side length $2 \ell$ centered at points $n_p$. We use $\ell = L^\alpha$,
where $(\alpha, \beta)$ satisfy \eqref{eq:constraints1}. For example, we can take $0 < \alpha < 1/20$. The number of such cubes $\Lambda_\ell (n_p)$ is $N_L := [ (2L+1) /(2 \ell+1) ]^d$. The local Hamiltonian is $H^\omega_{p, \ell}$. The associated eigenvalue point process is denoted by $\eta_{\ell,p}^\omega$. We define the point process $\zeta^\omega_{\Lambda_L} = \sum_{p=1}^{N_L} \eta^\omega_{p, \ell}$.
For a bounded interval $I \subset \R$, we define the local random variable $\eta^\omega_{\ell,p}(I) := {\rm Tr} ( E_{H^\omega_{p,\ell}}(I_L(E)))$ and similarly for the scaled interval $J_L(E^\prime)$.
For $p \neq p'$, these random variables are independent.
We compute
\bea\label{eq:independent1}
\P \{ (\zeta^\omega_{\Lambda_L}(I) \geq 1) \cap   (\zeta^\omega_{\Lambda_L}(J) \geq 1) \} &=& \sum_{p,p^\prime = 1}^{N_L} \P \{ (\eta^\omega_{\ell,p}(I) \geq 1) \cap   (\eta^\omega_{\ell,p}(J) \geq 1) \} \nonumber \\
 &= & \sum_{p, p^\prime = 1}^{N_L} \P \{ \eta^\omega_{\ell,p}(I) \geq 1 \} \P \{  \eta^\omega_{\ell,p}(J) \geq 1 \} \nonumber \\
 & & + \mathcal{E}_L(E, E^\prime, I, J) ,
 \eea
where the error term is just the diagonal $p=p'$ contribution:
\bea\label{eq:error1}
\mathcal{E}_L(E, E^\prime, I, J)  & = &   \sum_{p = 1}^{N_L} \left[ \P \{ (\eta^\omega_{\ell,p}(I) \geq 1) \cap
 (\eta^\omega_{\ell,p}(J) \geq 1) \} \right.             \nonumber \\
 &  & \left. -
 \P \{ \eta^\omega_{\ell,p}(I) \geq 1 \} \P \{ \eta^\omega_{\ell,p}(J) \geq 1 \} \right].
 \eea

The first probability on the right side of \eqref{eq:error1} is bounded above by $C_0 L^{- 2d(1 - 2 \beta - \alpha)}$
 due to the decorrelation estimate \eqref{eq:prob1}.
% Both probabilities are bounded by
%$C_0 (\ell/L)^{2d}$. The bound on the first probability on the right in \eqref{eq:error1} is
%a consequence of the decorrelation estimate.
The bound on the second probability on the right of \eqref{eq:error1} is $C_W^2 L^{-2d (1-\alpha)}$. It is obtained from the square of the Wegner estimate
$$
\P \{ \eta_{\ell,E^\prime}^{(p)}(J) \geq 1 \} \leq C_W (\ell / L)^d = C_W L^{-d(1 - \alpha)}.
$$
Since $N_L \sim (L / \ell)^d = L^{(1- \alpha)d}$, we find that
\beq\label{eq:error2}
  \mathcal{E}_L(E, E^\prime, I, J)  \leq C_W^2 L^{-d(1-\alpha)} + C_0 L^{-d(1 -  \alpha - 4 \beta)}  \rightarrow 0, ~~~ L \rightarrow \infty,
  \eeq
because of \eqref{eq:constraints1}.
Since the set of limit points $\zeta^\omega$ and $\xi^\omega$ are the same \cite{hislop-krishna1}, this estimate proves that
\beq\label{eq:independent2}
\lim_{L \rightarrow \infty}   \P \{ (\zeta^\omega_{E,\Lambda_L}(I) \geq 1) \cap   (\zeta^\omega_{E^\prime,\Lambda_L}(J) \geq 1) \} = \P \{ \xi^\omega_E(I) \geq 1 \} \P \{ \xi^\omega_{E^\prime} (J) \geq 1 \} ,
\eeq
establishing the asymptotic independence of the random variables $\xi^\omega_E(I)$ and $\xi^\omega_{E^\prime}(J)$ provided $|E - E^\prime | > 4d$.
%%%%%%%%%%%%%%%%%%%%%%%%%%%%%%%%%%%%%%%%%%%%%%%%%%%%%%%%%%%%%%%%%%%%%%%%%%%%%%%%%%%%%%%%%%

\section{Bounds on eigenvalue multiplicity}\label{sec:multiplicity1}

The extended Minami estimate may be used with the Klein-Molchanov argument \cite{klein-molchanov} to bound the multiplicity of eigenvalues in the localization regime. The basic argument of Klein-Molchanov is the following. If $H_\omega$ has at least $m_k +1$ linearly independent eigenfunctions with eigenvalue $E$ in the localization regime, so that the eigenfunctions exhibit rapid decay, then any finite volume operator $H_{\omega, L}$ must have at least $m_k + 1$ eigenvalues close to $E$ for large $L$. But, by the extended Minami estimate, this event occurs with small probability.
The first lemma is a deterministic result based on perturbation theory.

\begin{lemma}\label{lemma:local-ef-est1}
Suppose that $E \in \sigma (H)$ is an eigenvalue of a self adjoint operator $H$ with multiplicity at least $m_k + 1$. Suppose that all the associated eigenfunctions decay faster than $\langle x \rangle^{- \sigma}$, for some $\sigma > d / 2 > 0$. We define $\epsilon_L := C L^{- \sigma + \frac{d}{2}}$. Then for all $L>>0$, the local Hamiltonian $H_{L} := \chi_{\Lambda_L} H \chi_{\Lambda_L}$ has at least $m_k + 1$ eigenvalues in the interval $[E-  \epsilon_L, E+  \epsilon_L]$.
\end{lemma}

\begin{proof}
1. Let $\{ \varphi_j ~|~ j = 1, \ldots , M \}$ be an orthonormal basis of the eigenspace for $H$ and eigenvalue $E$. We assume that the eigenvalue multiplicity $M \geq m_k + 1$. We define the local functions $\varphi_{j , L} := \chi_{\Lambda_L} \varphi_j$, for $j = 1 , \ldots, M$.
These local functions satisfy:
\bea\label{eq:local-ef1}
1 - \epsilon_L & \leq & \| \varphi_{j,L}  \|  \leq   1 , \nonumber \\
 | \langle \varphi_{i,L} , \varphi_{j, L} \rangle | & \leq & \epsilon_L, ~~~~i \neq j .
\eea
It is easy to check that these conditions imply that the family is linearly independent. Let $V_L$ denote the $M$-dimensional subspace of $\ell^2 (\Lambda_L)$ spanned by these functions.

\noindent
2. As in \cite{klein-molchanov}, it is not difficult to
prove that the functions $\varphi_{j,L}$ are approximate eigenfunctions for $H_L$:
\beq\label{eq:approx-ef1}
\| (H_L - E) \varphi_{j,L} \| \leq \epsilon_L \| \varphi_{j,L} \| .
\eeq
Furthermore, for any $\psi_L \in V_L$, we have $\| (H_L - E) \psi_L \| \leq 2 \epsilon_L \| \psi_L \|$.

\noindent
3. Let $J_L := [ E - 3 \epsilon_L , E+ 3 \epsilon_L ]$. We write $P_L$ for the spectral projector $P_L :=
\chi_{J_L}(H_L)$ and $Q_L := 1 - P_L$ is the complementary projector. For any $\psi \in V_L$, we have
$\| Q_L \psi \| \leq (3 \epsilon_L)^{-1} \| (H_L - E) Q_L \psi \| \leq (2/3) \| \psi \|.$
Since $\|P_L \psi \|^2 = \| \psi \|^2 - \| Q_L \psi \|^2 \geq (5/9) \| \psi \|$,
it follows that $P_L: V_L \rightarrow \ell^2 (\Lambda_L)$ is injective. Consequently, we have
$$
{\rm dim} \Ran P_L = {\rm Tr} (P_L ) \geq {\rm dim} ~V_L = M > m_k.
$$
Redefining the constant $C >0$ in the definition of $\epsilon_L$, we find that
$H$ has at least $m_k + 1$ eigenvalues in $[E- \epsilon_L, E+ \epsilon_L]$.
\end{proof}

The second lemma is a probabilistic one and the proof uses the extended Minami estimate.

\begin{lemma}\label{lemma:local-ef-est2}
Let $I \subset \R$ be a bounded interval. For $q > 2d$,  and any interval
$J \subset I$ with $|J| \leq L^{-q}$, we define the event
\beq\label{eq:event1}
\mathcal{E}_{L,I,q} := \{ \omega ~|~ {\rm Tr} ( \chi_{J} (H_{\omega, L})) \leq m_k ~\forall J \subset I, |J| \leq L^{-q} \} .
\eeq
Then, the probability of this event satisfies
% for all $J \subset I$ with $|J| \leq L^{-q}$, we have
\beq\label{eq:prob1-1}
\P \{ \mathcal{E}_{L,I,q} \} \geq 1 - C_0 L^{2d-q} .
\eeq
\end{lemma}

\begin{proof}
 We cover the interval $I$ by $2 ( [ L^q |I| /2] +1)$ subintervals of length $2 L^{-q}$ so that any subinterval $J$ of length $L^{-q}$ is contained in one of these. We then have
\beq\label{eq:minami-prob1}
\P \{ \mathcal{E}_{L,I,q}^{\rm c} \} \leq (L^q |I| + 2) \P  \{ \chi_J(H_{\omega,L}) > m_k \}.
\eeq
The probability on the right side is estimated from the extended Minami estimate
\beq\label{eq:minami-prob2}
  \P  \{ \chi_J(H_{\omega,L}) > m_k \} \leq C_M (L^{-q} L^d)^2 = C_M L^{2(d-q)},
  \eeq
so that
\beq\label{eq:minami-prob3}
\P \{ \mathcal{E}_{L,I,q}^{\rm c} \} \leq C_M (L^q |I| + 2) L^{2(d-q)} = C_M (|I| +1)L^{2d-q}.
\eeq
This establishes \eqref{eq:prob1-1}.
\end{proof}

\begin{thm}\label{thm:multiplicity1}
Let $H^\omega$ be the generalized Anderson Hamiltonian described in section \ref{sec:introduction} with perturbations $P_i$ having uniform rank $m_k$.  Then the eigenvalues in the localization regime have multiplicity at most $m_k$ with probability one.
\end{thm}

\begin{proof}
We consider a length scale $L_k = 2^k$. It follows from \eqref{eq:prob1-1} that
the probability of the complementary event $\mathcal{E}_{L_k,I,q}^{\rm c}$
is summable. By the Borel-Cantelli Theorem, that means  for almost every $\omega$ there is a $k(q, \omega)$ so that for all $k > k(q, \omega)$
the event $\mathcal{E}_{L_k,I,q}$ occurs with probability one. Let us suppose that $H^\omega$ an eigenvalue with multiplicity at least $m_k +1$ in an interval $I$ and that the corresponding eigenfunctions decay exponentially. Then, by Lemma \ref{lemma:local-ef-est1}, the local Hamiltonian $H_{\omega, L_k}$ has at least $m_k + 1$ eigenvalues in the interval $[E- \epsilon_L, E+ \epsilon_L]$ where $\epsilon_L = CL^{-(\beta - \frac{d}{2})}$, for any $\beta > 5 d/2$.  This contradicts the event $\mathcal{E}_{L_k,I,q}$ which states that there are no more than $m_k$ eigenvlaues in any subinterval $J \subset I$ with $|J| \leq L^{-q}$ since we can find $q > 2d$ so that $\beta - \frac{q}{2} > q$.
\end{proof}

It appears that the simplicity of eigenvalues in the localization regime might be enough to imply a Minami estimate.
Further investigations on the simplicity of eigenvalues for Anderson-type models may be found in the article by Naboko, Nichols, and Stolz \cite{naboko-nichols-stolz}

\end{document}